%% file: main.tex
\RequirePackage[l2tabu, orthodox]{nag}
\documentclass[a4paper, 11pt]{article}

\usepackage{lmodern}
\usepackage[utf8]{inputenc}
\usepackage[T1]{fontenc}
\usepackage{microtype}

\usepackage[a4paper, top=1in, bottom=1in, left=1in, right=1in]{geometry}

\usepackage{cite}

\usepackage{graphicx}
\usepackage{amsmath}
\usepackage{amsfonts}

\usepackage{amsthm}
\newtheorem{lemma}{Lemma}
\newtheorem{observation}{Observation}
\newtheorem{theorem}{Theorem}

\newtheorem{definition}{Definition}
\newtheorem{remark}{Remark}

\usepackage{thmtools}
\usepackage{thm-restate}

\usepackage[ruled,vlined,linesnumbered]{algorithm2e}

\usepackage{authblk}

\newcommand{\E}{\mathbb{E}}

\usepackage{xspace}

\DeclareMathOperator{\rank}{rank}
\DeclareMathOperator{\pos}{pos}
\DeclareMathOperator{\disl}{disl}
\DeclareMathOperator{\score}{score}

\newcommand{\merge}{\texttt{Merge}\xspace}
\newcommand{\windowsort}{\texttt{WindowSort}\xspace}
\newcommand{\rifflesort}{\texttt{RiffleSort}\xspace}

\title{Optimal Sorting with Persistent Comparison Errors\thanks{Research supported by SNF (project number 200021\_165524).}}
\author{Barbara Geissmann}
\author{Stefano Leucci}
\author{Chih-Hung Liu}
\author{Paolo Penna}

\affil{Department of Computer Science, ETH Zürich. 
\makebox[\textwidth][c]{\texttt{\{barbara.geissmann,stefano.leucci,chih-hung.liu,paolo.penna\}@inf.ethz.ch}}}
\date{}

\begin{document}

\maketitle
\pagestyle{empty}

\begin{abstract}
	We consider the problem of  sorting $n$ elements in the case of \emph{persistent} comparison errors. In this model (Braverman and Mossel, SODA'08), each comparison between  two elements can be wrong with some fixed (small) probability $p$, and \emph{comparisons cannot be repeated}. 
	Sorting perfectly in this model is impossible, and the objective is to minimize the \emph{dislocation} of each element in the output sequence, that is, the difference between its true rank and its position. Existing lower bounds for this problem show that no algorithm can guarantee, with high probability, \emph{maximum dislocation} and \emph{total dislocation} better than $\Omega(\log n)$ and $\Omega(n)$, respectively, regardless of its running time.
	
	In this paper, we present the first \emph{$O(n\log n)$-time} sorting algorithm that guarantees both \emph{$O(\log n)$ maximum dislocation} and \emph{$O(n)$ total dislocation} with high probability. Besides improving over the previous state-of-the art algorithms -- the best known algorithm had running time $\tilde{O}(n^{3/2})$ --  our result indicates that comparison errors do not make the problem computationally more difficult: a sequence with the best possible dislocation can be obtained in $O(n\log n)$ time and, even without comparison errors, $\Omega(n\log n)$ time is necessary to guarantee such dislocation bounds.

	In order to achieve this optimal result, we  solve two sub-problems, and the respective methods have their own merits for further application.
	One is how to locate a position in which to insert an element in an almost-sorted sequence having $O(\log n)$ maximum dislocation in such a way that the dislocation of the resulting sequence will still be $O(\log n)$.

	The other is how to simultaneously insert $m$ elements into an almost sorted sequence of $m$ different elements, such that the resulting sequence of $2m$ elements remains almost sorted. 
\end{abstract}

\clearpage
\pagestyle{plain}
\setcounter{page}{1}

\input{Introduction}

\section{Preliminaries}

\label{sec:preliminaries}

According to our error model, elements possess a true total linear order, however this order can only be observed through noisy comparisons. In the following, given two distinct elements $x$ and $y$, we will write $x \prec y$ (resp. $x \succ y$) to 
mean that $x$ is smaller (resp. larger) than $y$ according to the true order, and $x<y$ (resp. $x>y$) to mean that $x$
appears to be smaller (resp. larger) than $y$ according to the observed comparison result.

Given a sequence or a set of elements $A$ and an element $x$ (not necessarily in $A$), we define $\rank(x, A) = |\{ y \in A : y \prec x\}|$ be as the \emph{true rank} of element $x$ in $A$ (notice  that ranks start from $0$).
Moreover, if $A$ is a sequence and $x \in A$, we denote by $\pos(x, A) \in [0, |S'|-1]$ the \emph{position} of $x$ in $A$ (notice that positions are also indexed from $0$), so that the \emph{dislocation} of $x$ in $A$ is
$\disl(x,S) = |\pos(x,S)-\rank(x,S)|$, and the \emph{maximum dislocation} of the sequence $A$ is $\disl(S) = \max_{x \in S} \disl(x,S)$.

\noindent For $z \in \mathbb{R}$, we write $\ln z$ and $\log z$ to refer to the natural and the binary logarithm of $z$, respectively.

\section{Noisy Binary Search}
\label{sec:noisy_binary_search}

Given a sequence $S = \langle s_0, \dots, s_{n-1} \rangle$ of $n$ elements with maximum dislocation $d \ge \log n$, and element $x$ not in the sequence, we want to compute in time $O(\log n)$ an \emph{approximate rank} of $x$ in $S$, that is, a position where to insert $x$ in $S$ while preserving a $O(d)$ upper bound on dislocation of the resulting sequence. 
More precisely, we want to compute index $r_x$ such that $|r_x - \rank(x,S)|=O(d)$, in presence of persistent comparison errors: Errors between $x$ and the elements in $S$ happen independently with probability $p$, and whether the comparison $x$ between x and an element $y \in S$ is correct or erroneous does not depend on the position of $y$ in $S$, nor on the actual permutation of the sorted elements induced by their order in $S$ (i.e., we are not allowed to pick the order of the elements in $S$ as a function of the errors). We do not impose any restriction on the errors for comparisons that do not involve $x$.

In the following, we will show an algorithm that computes such a rank $r_x$ in time $O(\log n)$.
This immediately implies that $O(\log n)$ time also suffices to insert $x$ into $S$
so that the resulting sequence $\langle s_0, \dots, s_{r_x-1}, x, s_{r_x}, s_{n-1}\rangle$ still has maximum dislocation $O(d)$.
\begin{remark}
Notice that the $O(\log n)$ running time is asymptotically optimal for all $d=n^{1-\epsilon}$,  for constant  $\epsilon<1$,  since a $\Omega(\log n - \log d) = \Omega(\log n)$ decision-tree lower bound holds even in absence of comparison errors.
\end{remark}

In the following, for the sake of simplicity, we let $c = 10^3$ and we assume that $n = 2 c d \cdot 2^h - 1$ for some non-negative integer $h$. Moreover, we focus on $p \le \frac{1}{32}$ even though this restriction can be easily removed to handle all $p < \frac{1}{2}$, as we argue at the end of the section.

We consider the set $\{0, \dots, n\}$ of the possible ranks of $x$ in $S$
and we subdivide them into
$2 \cdot 2^h$ ordered \emph{groups} $g_0, g_1, \dots$ each containing $cd$ contiguous positions, namely, group $g_i$ contains positions  $c i d$, \dots, $c (i+1) d -1$.
Then, we further partition these $2 \cdot 2^h$ groups into two ordered sets $G_0$ and $G_1$, where $G_0$ contains the groups $g_i$ with even $i$ ($i \equiv 0 \pmod{2}$) and $G_1$ the groups $g_i$ with odd $i$ ($i \equiv 1 \pmod{2}$). Notice that $|G_0| =|G_1|= 2^h$. In the next section, for each $G_j$, we shall define a \emph{noisy binary search tree} $T_j$, which will be the main ingredient of our algorithm.

\subsection{Constructing $T_0$ and $T_1$}

Let us consider a fixed $j \in \{0,1\}$ and define $\eta = 2 \lceil \log n \rceil$. 
The tree $T_j$ comprises of a binary tree of height $h + \eta$ in which the first $h+1$ levels (i.e., those containing vertices at depths $0$ to $h$) are complete and the last $\eta $ levels consists of $2^h$ paths of $\eta$ vertices, each emanating from a distinct vertex on the $(h+1)$-th level.
We index the leaves of the resulting tree from $0$ to $2^h-1$,
we use $h(v)$ to denote the depth of vertex $v$ in $T_j$,
and we refer to the vertices $v$ at depth $h(v) \ge h$  as \emph{path-vertices}.
Each vertex $v$ of the tree is associated with one \emph{interval} $I(v)$, i.e., as a set of contiguous positions, as follows: for a leaf $v$ having index $i$, $I(v)$ consists of the positions in $g_{2i+j}$; for a non-leaf path-vertex $v$ having $u$ as its only child, we set $I(v)=I(u)$; finally, for an internal vertex $v$ having $u$ and $w$ as its left and right children, respectively, we define $I(v)$ as the interval containing all the positions between $\min I(u)$ and $\max I(w)$ (inclusive).

\begin{figure}
	\centering
    \includegraphics[width=.9\textwidth]{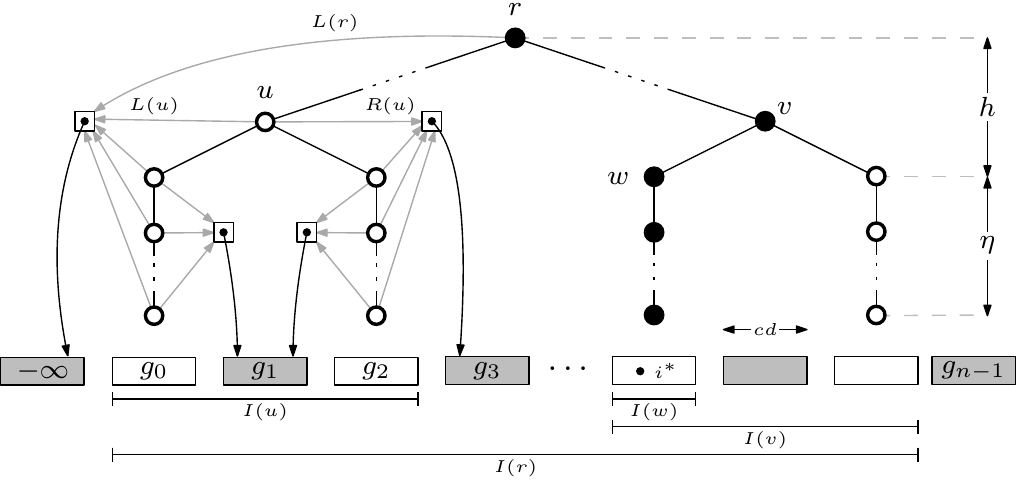}
    \caption{An example of the noisy tree $T_0$. On the left side the shared pointers $L(\cdot)$ and $R(\cdot)$ are shown. Notice how $L(r)$ (and, in general, all the $L(\cdot)$ pointers on the leftmost side of the tree) points to the special $-\infty$ element.  Good vertices are shown in black while bad vertices are white. Notice that, since $i^* \in I(w)$, we have $T^* = T_0$ and hence all the depicted vertices are either good or bad.}
    \label{fig:noisy_tree}
\end{figure}

Moreover, each vertex $v$ of the tree has a reference to two \emph{shared pointers} $L(v)$ and $R(v)$ to positions in $\{0, \dots, n\} \setminus \bigcup_{g_i \in S_j} g_i$. Intuitively, $L(v)$ (resp. $R(v)$) will always point to positions of $S$ occupied by elements that are \emph{smaller} (resp. \emph{larger}) than all the elements $s_i$ with $i \in I(v)$.
For each leaf $v$, let $L(v)$ initially point to $\min I(v) - d - 1$ and $R(v)$ initially point to $\max I(v) + d$.
A non-leaf path-vertex $v$ shares both its pointers with the corresponding pointers of its only child, while a non-path vertex $v$ shares its left pointer $L(v)$ with the left pointer of its left child, and its right pointer $R(v)$ with the right pointer of its right child. See Figure~\ref{fig:noisy_tree} for an example.

Notice that we sometimes allow $L(v)$ to point to negative positions and $R(v)$ to point to positions that are larger than $n-1$. In the following we consider all the elements $s_i$ with $i < 0$ (resp. $i \ge n)$ to be copies a special $-\infty$ (resp. $+\infty$) element such that $-\infty \prec x$ and $-\infty < x$ in every observed comparison (resp. $+\infty \succ x$ and $+\infty > x$).

\subsection{Walking on $T_j$}

The algorithm will perform a discrete-time random walk on each $T_j$.
Before describing such a walk in more detail, it is useful to define the following operation: 
\begin{definition}[test operation]\label{def:test}
	A \emph{test} of an element $x$ with a vertex $v$ is performed by (i) comparing $x$ with the elements $s_{L(v)}$ and $s_{R(v)}$, (ii) decrementing $L(v)$ by $1$ and, (iii) incrementing $R(v)$ by $1$. The tests succeeds if the observed comparison results are $x > s_{L(v)}$ and $x < s_{R(v)}$, otherwise the test fails.
\end{definition}

The walk on $T_j$ proceeds as follows. At time $0$, i.e., before the first step, the \emph{current} vertex $v$ coincides with the root $r$ of $T_j$.
Then, at each time step, we \emph{walk} from the current vertex $v$ to the next vertex as follows:
\begin{enumerate}
	\item We test $x$ with all the children of $v$ and, if \emph{exactly one} of these tests succeeds, we\emph{ walk to the corresponding child}.
	\item Otherwise, if \emph{all the tests fail}, we\emph{ walk to the parent} of $v$, if it exists.
\end{enumerate} 
In the remaining cases, we ``walk'' from $v$ to itself.

\smallskip
\noindent
We fix an upper bound  $\tau = 240 \lfloor \log n \rfloor$ on the total number of steps we perform. The walk stops as soon as one of the following two conditions is met: 
\begin{description}
	\item[Success:] The current vertex $v$ is a leaf of $T_j$, in which case we return $v$;
    \item[Timeout:] The $\tau$-th time step is completed and the success condition is not met.
\end{description}

\subsection{Analysis}

Let $i^* = \rank(x, S)$, and let 
$T^*$ be the unique tree in $\{ T_0, T_1 \}$ such that
$i^*$ belongs to the interval of a leaf in $T^*$, and let $T'$ be the other tree.

\begin{definition}[good/bad vertex]
	We say that a vertex  $v$ of $T^*$ is \emph{good} if $i^* \in I(v)$ and \emph{bad} if either $i^* < \min I(v) - cd$ or $i^* > \max I(v) + cd$.
\end{definition}

Notice that in $T^*$ all the vertices are either good or bad, that the intervals corresponding to vertices at the same depth in $T^*$ are pairwise disjoint, and that the set of good vertices is exactly a root-to-leaf path.
Moreover, all path-vertices of $T'$ are either bad or neither good nor bad.
In both $T'$ and $T^*$ all the children of a bad vertex are also bad. 

\begin{lemma}
\label{lemma:test_good}
A test on a good vertex succeeds with probability at least $1-2p$.
\end{lemma}
\begin{proof}
	Let $v$ be a good vertex.
    Since $i^* \in I(v)$ and the pointer $L(v)$ only gets decremented, we have $L(v) \le \min I(v) - d - 1 \le i^* -d - 1$. Since $S$ has dislocation at most $d$,  $\rank(s_{L(v)}, S) \leq L(v) + d \le  i^*-1$.
	This implies that $s_{L(v)} \prec x$ and hence the probability to observe $S_{L(v)} > x$  during the test is at most $p$.

	Similarly, $R(v)$ is only incremented during the walk and hence $R(v) \ge \max I(v) + d \ge i^* + d$.
    Since $S$ has dislocation at most $d$, this implies that $\rank(s_{R(v)}, S) \ge R(v)-d \ge i^*$ and, in turn, that $s_{R(v)} \succ x$. Therefore, $s_{R(v)} < x$ is observed with probability at most $p$.

By the union bound, $s_L(v) < x < s_R(v)$ with probability at least $1-2p$.
\end{proof}

\begin{lemma}
\label{lemma:test_bad}
A test on a bad vertex succeeds with probability at most $p$.
\end{lemma}
\begin{proof}
	Let $v$ be a bad vertex and notice that at most $\tau - 1 \le 240  \log n  - 1 \le 240 d - 1 < (c-2)d - 1$ tests have been performed before the current test. 
    We have that either $i^* < \min I(v) - cd$ or $i^* > \max I(v) + cd$.
    
    In the former case, we have $L(v) \ge \min I(v) - d - 1 - (\tau -1) \ge \min I(v) - d - (c-2)d > (i^* + cd) - (c-1)d = i^* + d$. Since $S$ has dislocation at most $d$, this implies  $\rank(s_{L(x)}, S) \geq L(x) -d \ge i^*$. Therefore, $s_{L(x)} \succ x$ and thus $s_{L(x)} > x$ is also observed with probability at least $1-p$, causing the test to fail.

	Similarly, in the latter case, we have $R(v) \le \max I(v) + d + (\tau-1) \le \max I(v) + d + (c-2)d - 1 < (i^* - cd) + (c-1)d - 1 = i^* - d - 1$.
	Therefore, $\rank(s_{R(v)}, S) \le i^* - 1$, implying $s_{R(v)} \prec x$, and hence $s_{R(v)} < x$ is observed with probability at least $1-p$, causing the test to fail.    
	\end{proof}

\noindent We say that a step on $T_j$ from vertex $v$ to vertex $u$  is \emph{improving} iff either:
\begin{itemize}
	\item $u$ is a good vertex and $h(u) > h(v)$ (this implies that $v$ is also good); or
    \item $v$ is a bad vertex and $h(u) < h(v)$.
\end{itemize}
Intuitively, each improving step is making progress towards identifying the interval containing the true rank $i^*$ of $x$, while each non-improving step undoes the progress of at most one improving step.

\begin{lemma}
	\label{lemma:improving_step}
	Each step performed during the walk on $T^*$ is improving with probability at least $1-3p$.
\end{lemma}
\begin{proof}
		Consider a generic step performed during the walk on $T^*$ from a vertex $v$.
        If $v$ is a good vertex, then exactly one child $u$ of $v$ in $T^*$ is good (notice that $v$ cannot be a leaf). 
        By Lemma~\ref{lemma:test_good}, the test on $u$ succeeds with probability at least $1-2p$. Moreover, there can be at most one other child $w \neq u$ of $v$ in $T^*$. If $w$ exists, then it must be bad and, by Lemma~\ref{lemma:test_bad}, the test on $w$ fails with probability at least $1-p$.
	By using the union bound on the complementary probabilities, we have that process walks from $v$ to $u$ with probability at least $1-3p$.

	If $v$ is a bad vertex, then all of its children are also bad. Since $v$ has at most $2$ children and, by Lemma~\ref{lemma:test_bad}, a test on a bad vertex fails with probability at least $1-p$, we have that all the tests fail with probability at least $1-2p$. In this case the process walks from $v$ to the parent of $v$ (notice that, since $v$ is bad, it cannot be the root of $T^*$).
\end{proof}

Since the walk on $T_j$ takes at most $\tau$ steps, any two distinct shared pointers $L(v)$ and $R(u)$ initially differ by at least $cd - 2d - 1 > (c-3)d$ positions, each step increases/decreases at most $4$ pointers (as it performs at most $2$ tests), and $\tau = 240 \lfloor \log n \rfloor < \frac{c-3}{4} d$, we can conclude that no two distinct pointers will ever point to the same position. This, in turn, implies:
\begin{observation}
	\label{obs:independent_comparisons}
	Element $x$ is compared to each element in $S$ at most once.
\end{observation}

The following lemmas show that the walk on $T^*$ is likely to return a vertex corresponding to a good interval, while the walk on $T' \neq T^*$ is likely to either timeout or to return a non-bad vertex, i.e., a vertex whose corresponding interval contains positions that are close to the true rank of $x$ in $S$.

\begin{lemma}
	\label{lemma:timeout}
	The walk on $T^*$ timeouts with probability at most $e \cdot n^{-6}$.
\end{lemma}
\begin{proof}
	For $t=1,\dots, \tau$, let $X_t$ be an indicator random variable that is equal to $1$ iff the $t$-th step of the walk on $T^*$ is improving. If the $t$-th step is not performed then let $X_t = 1$.

	Notice that if, at any time $t'$ during the walk, the number $X^{(t')} = \sum_{t=1}^{t'} X_t$ of improving steps exceeds the number of non-improving steps by at least $h + \eta$, then the success condition is met.
    This means that a necessary condition for the walk to timeout is $X^{(\tau)} - (\tau - X^{(\tau)}) < h + \eta$, which is equivalent to $X^{(\tau)}  < \frac{h + \eta + \tau}{2}$.
    
    By Observation~\ref{obs:independent_comparisons} and by Lemma~\ref{lemma:improving_step}, we know that 
    each $X_t$s corresponding to a performed steps satisfies $P(X_t = 1) \ge 1-3p > 1 - \frac{1}{10}$ (since $p \le \frac{1}{32}$), regardless of whether the other steps are improving.
    We can therefore consider the following experiment: at every time step $t=1,\dots, \tau$ we flip a coin that is heads with probability $q = \frac{9}{10}$, we let $Y_t = 1$ if this happens and $Y_t = 0$ otherwise.
    Clearly, the probability that $X^{(\tau)}  < \frac{h + \eta + \tau}{2}$ is at most the probability that 
    $Y^{(\tau)}  = \sum_{t=1}^{\tau} Y_t < \frac{h + \eta + \tau}{2}$. 
    By noticing that $\tau > 3 (h + \eta )$, and by using the Chernoff bound:
    \begin{multline*}
    	\Pr\left(X^{(\tau)} < \frac{h + \eta + \tau}{2}\right)
        \le \Pr\left(Y^{(\tau)} < \frac{h + \eta + \tau}{2}\right)
        \le \Pr\left(Y^{(\tau)} < \frac{2\tau}{3}\right) 
        = \Pr\left(Y < \frac{2 \E[Y]}{3q} \right) \\
        < \Pr\left(Y < \left(1- \frac{1}{4}\right) \E[Y] \right)
        \le e^{-\frac{\tau q}{32}}
        <  e^{-\frac{\tau}{40}} 
        \le e^{1-6 \log n} < e \cdot n^{-6}.
    \end{multline*}
\end{proof}

\begin{lemma}
	\label{lemma:bad_vertex_returned}
	A walk on $T_j$ returns a bad vertex with probability at most $240 n^{-7}$.
\end{lemma}
\begin{proof}
	Notice that, in order to return a bad vertex $v$, the walk must
    first reach a vertex $u$ at depth $h(u)=h$, and then traverse the $\eta$ vertices of the path rooted in $u$ having $v$ as its other endpoint.
        
    We now bound the probability that, once the walk reaches $u$, it will also reach $v$ before walking back to the parent of $u$.
    Notice that all the vertices in the path from $u$ to $v$ are associated to the same interval, and hence they are all bad. 
	Since a test on a bad vertex succeeds with probability at most $p$ (see Lemma~\ref{lemma:test_bad}) and tests are independent (see Observation~\ref{obs:independent_comparisons}), the sought probability can be upper-bounded by considering a random walk on $\{0 ,\dots, \eta+1 \}$ that: (i) starts from $1$, (ii) has one absorbing barrier on $0$ and another on $\eta+1$, and (iii) for any state $i \in [1, \eta]$ has a probability of transitioning to state $i+1$ of $p$ and to state $i-1$ of $1-p$.
    Here state $0$ corresponds to the parent of $u$, and state $i$ for $i>0$ corresponds to the vertex of the $u$--$v$ path at depth $h+i-1$ (so that state $1$ corresponds to $u$ and state $\eta+1$ corresponds to $v$).
    
    The probability of reaching $v$ in $\tau$ steps is at most the probability of being absorbed in $\eta$ (in any number of steps), which is (see, e.g., \cite[pp.~344--346]{feller1957introduction}):
\[
		\frac{  \frac{1-p}{p} - 1 }{ (\frac{1-p}{p})^{\eta+1} - 1 } \le
    	\left( \frac{p}{1-p} \right)^\eta \le
    	\left( \frac{1}{31} \right)^{2\log n}  <
        \left( \frac{1}{2^4} \right)^{2\log n} 
        = 2^{-8 \log n} = n^{-8},
\]
where we used the fact that $p \le \frac{1}{32}$.

Since the walk on $T_j$ can reach a vertex at depth $h$ at most $\tau$ times, by the union bound we have that the probability of returning a bad interval is at most $\tau n^{-8} \le 240 n^{-8}  \log n  \le 240 n^{-7}$.
\end{proof}

\noindent We are now ready to prove the main result of this section.

\begin{theorem}
	\label{thm:noisy_binary_search}
	Let $S$ be an sequence of $n$ elements having maximum dislocation at most $d \ge \log n$ and let $x \not\in S$.
	Under our error model, an index $r_x$ such that $r_x \in [ \rank(x,S) - \alpha d,  \rank(x,S) + \alpha d]$ can be found in $O(\log n)$ time with probability at least $1- O(n^{-6})$, where $\alpha > 1$ is an absolute constant.
\end{theorem}
\begin{proof}
	We compute the index $r_x$ by performing two random walks on $T_0$ and on $T_1$, respectively.
    If any of the walks returns a vertex $v$, then we return any position in the interval $I(v)$ associated with $v$. If both walks timeout we return an arbitrary position.
    
    From Lemma~\ref{lemma:timeout} the probability that both walks timeout is at most $\frac{e}{n^6}$ (as the walk on $T^*$ timeouts with at most this probability).
    Moreover, by Lemma~\ref{lemma:bad_vertex_returned}, the probability that at least one of the two walks returns a bad vertex is at most $\frac{480}{n^7}$.
    
    By the union bound, we have that with probability at most $1-\frac{480}{n^7} - \frac{e}{n^6} = 1 - O(\frac{1}{n^6})$, vertex $v$ exists and it is not bad. In this case, using the fact that $r_x \in [\min I(v), \max I(v)]$ and that $\max I(v) - \min I(v) < cd$, we have:
    \[
	    i^* \ge \min I(v) - cd  > \max I(v) - 2cd \ge r_x - 2cd,
    \]
    and
	\[
    	i^* \le \max I(v) + cd < \min I(v) + 2cd \le r_x + 2cd.
     \]
     
     To conclude the proof it suffices to notice that the random walk requires at most $\tau = O(\log n)$ steps, that each step requires constant time, and that it is not necessary to explicitly construct $T_0$ and $T_1$ beforehand.
     Instead, it suffices to maintain a partial tree consisting of all the vertices visited by the random walk so far: vertices (and the corresponding pointers) are \emph{lazily} created and appended to the existing tree whenever the walk visits them for the first time.
\end{proof}

To conclude this section, we remark that our assumption that $p \le \frac{1}{32}$ can be easily relaxed to handle any constant error probability $p < \frac{1}{2}$. This can be done by modifying the test operation so that, when $x$ is tested with a vertex $v$, the majority result of the comparisons between $x$ and the set $\{ s_{L(v)}, s_{L(v)-1}, \dots, s_{L(v)-k+1} \}$ (resp.
$x$ and the set $\{ s_{R(v)}, s_{R(v)+1}, \dots, s_{R(v)+k-1} \}$) of $\eta$ elements is considered, where $k$ is a constant that only depends on $p$.
Consistently, the pointers $L(v)$ and $R(v)$ are shifted by $k$ positions, and the group size is increased to $k \cdot c$ to ensure that Observation~\ref{obs:independent_comparisons} still holds. Notice how our description for $p \le \frac{1}{32}$ corresponds exactly to the case $k=1$. 
The only difference in the statement Theorem~\ref{thm:noisy_binary_search} is that $\alpha$ is no longer an absolute constant, rather, it depends (only) on the value of $p$.

\input{Sorting.tex}

\section{WindowSort}\label{sec:windowsort}

To make our algorithm description self-contained, and in order to prove Theorem~\ref{thm:windowsort} (thus strengtening the result of\cite{geissmann_et_al}), Algorithm~\ref{alg:windowsort} reproduces (a slightly modified version of) the pseudocode of \windowsort.
\windowsort receives in input a sequence $S$ of $n$ elements, and an additional upper bound $d$ on the initial dislocation of $S$. The original \windowsort algorithm corresponds to the case $d=n$.

\begin{algorithm}[t]
	$S_{2d} \gets S$\;
	\ForEach{$w=2d, d, d/2, \dots, 2$}
	{
		Let $\langle x_0, \dots, x_{n-1} \rangle$ be the elements in $S_{w}$\;
		\ForEach{$x_i \in S_w$}
		{
			$\score_w(x_i) \gets \max\{0, i - w\} + \{ x_j < x_i \, : \, |j-i| \le w \}$
		}
		
		$S_{ w/2 } \gets $ sort the element of $S_{w}$ by non-decreasing value of $\score_w( \cdot )$\;
	}
	\Return $S_1$
	\caption{\windowsort\unskip($S,d$)}
	\label{alg:windowsort}
\end{algorithm}

\windowsort iteratively computes a collection $\{S_{2d},S_{d}, S_{d/2},\dots, S_1\}$ of permutations of $S$:
at every step, \windowsort maintains a \emph{window size} $w$ and builds $S_{w}$ from $S_{2w}$.
Intuitively, for the algorithm to be successful, we would like $S_w$ to be a permutation of $S$ having maximum dislocation at most $w/2$, w.h.p. Even though this is true in the beginning (since we initially set $w=2d$), this property only holds up to a certain window size $w^*=\Theta(\log n)$.
Nevertheless, it is still possible to show that the maximum dislocation of the returned sequence is $\Theta(\log n)$ and that the
expected dislocation of each element is constant.
We summarize the above discussion in the following lemma, which follows from the same arguments used in the proofs of Theorems~9 and 14 in \cite{geissmann_et_al}: 

\begin{lemma}
	\label{lemma:windowsort}
	Let $S$ be a sequence of $n$ element having maximum dislocation at most $d$. Then, with probability $1-\frac{1}{n^5}$, the following properties hold: (i) there exists a window size $w^* = \Theta(\log n)$ such that $\disl(S_{w^*})=O(\log n)$; and (ii) $\E[\disl(x,S_1)] = O(1)$. 
	All the hidden constants depend only on $p$.
\end{lemma}

In the following, we shall prove that the total dislocation of the sequence returned by \windowsort is linear with high probability.
We start by providing an upper bound to the change in position of an element between different iterations of \windowsort. 
This will also immediately imply that an element can only move by at most $O(w)$ positions between $S_{w^*}$ and $S_1$.

\begin{lemma}\label{lem:move}
	For every $x \in S$, $|  pos(x, S_{w}) - pos(x, S_{w'}) | \le 4 |w-w'|$.	
\end{lemma}
\begin{proof}
	Without loss of generality let $w' < w$ (the case $w'>w$ is symmetric, and the case $w'=w$ is trivial).
	Consider a generic iteration of \windowsort corresponding to window size $w'' < w$.
	Let $i=\pos(x, S_{w})$ and $\Delta_{w''} = |i - pos(x, S_{w''/2})|$.
	
	\noindent For every element $y$ such that $pos(y, S_{w''}) < i-2w''$ we have:
	\[
	\score_{w''}(x) \ge i - w'' > \pos(y, S_{w''}) + w'' \ge \score_{w''}(y), 
	\]
	implying that $pos(x, S_{w''/2}) \ge  i-2w''$.
	Similarly, for every element $y \in S$ such that $pos(y, S_{w''}) > i + 2w''$:
	\[
	\score_{w''}(x) \le i+w'' < \pos(y, S_{w''})-w'' \le \score_{w''}(y),
	\]
	showing that the position of $x$ in $S_{w''/2}$ is at most $i+2w''$.
	We conclude that $\Delta_{w''} \le 2w''$, which allows us to write:
	\begin{multline*}
	|i - pos(x_i, S_{w'}) | = \Delta_{w} + \Delta_{w/2} + \Delta_{w/4} + \dots + \Delta_{2w'} \\  
	\le 2w + w + w/2 + \dots + 4w' = 4w - (2w' + w' + w'/2 + \dots) 
	= 4w - 4w'. 	
	\end{multline*}
\end{proof}

The previous lemma also allows us to show that, once the window size $w$ is sufficiently small, the final position of an element only depends on a small subset of nearby elements. This, in turn, will imply that, once $S_w$ is fixed, the final positions of distant elements in $S_1$ are conditionally independent. The above property is formally shown in the following:

\begin{lemma}\label{lem:dependent}
	Let $x \in S$.
	Given $S_{w}$,  $pos(x, S_1)$ only depends on comparisons involving elements in positions $pos(x, S_1)-6w, \dots, pos(x, S_1)+6w$ in $S_{w}$.
\end{lemma}
\begin{proof}
	Let $w' = w/2^t$, for $t=1, 2, \dots$, and let $S_w'=\langle x_0,\dots,x_{n-1}\rangle$.
	We prove by induction on $t$, that $i=pos(x_i,S_{w'})$ only depends on the comparison results with elements in positions $i-r_t, \dots, i+r_t$ in $S_w$, where $r_t = 6 w (1 - 2^{-t} )$.
	
	Base case $t=1$: By Lemma~\ref{lem:move} we have that $ i - 2w \le pos(s_i, S_{w}) \le i + 2w$.
	Therefore, $x_i$ is can be only compared to elements $x_j$ such that $i - 3w \le pos(x_j, S_{w}) \le i + 3w$

	Suppose now that the claim is true for some $t \ge 1$.
	We show that the claim also holds for $t+1$.
	Indeed, by Lemma~\ref{lem:move}, we have that $ i - 2w/2^t \le pos(x_i, S_{w/2^t}) \le i + 2w/2^t$
	and hence it is only compared to $x_j$s such that 
	$i - 3w/2^t \le pos(x_j, S_{w/2^t}) \le i + 3 w/2^t$.	
	
	By induction hypothesis, these elements depend only on elements in positions
	$3w/2^t + 6w - 6w 2^{-t}
	= 6w - 3w/2^t
	= 6w (1-2^{-(t+1)})$.
\end{proof}

\noindent We are finally ready to prove Theorem~\ref{thm:windowsort} used in section \ref{sec:rifflesort}, that we restate here for convenience.

\windowsortthm*
\begin{proof}
	First of all notice that Lemma~\ref{lem:move}, ensures that the final dislocation of each element in $S_1$ will be at most $d + 4w$ where $w=2d$ is the initial window size of \windowsort, thus implying the $c_p \cdot d$ bound on the final maximum dislocation.

	We will condition on the event that for some $w^* = \Theta( \log n)$, $S_{w^*}$ has maximum dislocation $\delta = O(\log n)$. Let $G$ be the indicator random variable that describes this event, i.e., $G=1$ if the event happens. 
	Clearly, by Lemma~\ref{lemma:windowsort},
	we have that $P(G=1) \ge 1 - \frac{1}{n^5}$, implying that the final maximum dislocation will be at most $\delta + 4 w^* = O(\log n)$ with the same probability. Therefore, we now only focus on bounding the total dislocation of $S_1$.

	Let $S_1= \langle x_0,\dots,x_{n-1}\rangle $. 
 	and observe that $\E[\disl(x_i,S_1)\mid G=1] = O(1)$. Indeed, by Lemma~\ref{lemma:windowsort},
	$\E[\disl(x_i,S_1)] \le c$ for all $x_i \in S$ and some constant $c>0$ depending only on $p$. Therefore:
	\begin{align*}
	c & \ge \E[\disl(x_i, S_1)]  =  \E[\disl(x_i,S_1)\mid G=1]\cdot P(G=1) + \E[\disl(x_i,S_1)\mid G=0]\cdot P(G=0)\\
	&\ge \E[\disl(x_i,S_1)\mid G=1]\cdot \left(1-\frac{1}{n^5}\right),
	\end{align*}	
	implying that $\E[\disl(x_i,S_1)\mid G=1] \le \frac{32}{31}c$ for all $n \ge 2$.

	We partition the elements of $S$ into $k = 20w^* + 4\delta$ sets, $P^{(0)},\dots,P^{(k-1)}$, 
	such that $P^{(j)} = \{x \in S \, : \, \rank(x, S) = j \pmod{k} \}$.
	We now show that the final positions of two elements belonging to the same set are conditionally independent on $G=1$.
	Indeed, assuming $G=1$ and using Lemma~\ref{lem:move}, we have:
	\[
		|\pos(x, S_1) - \rank(x, S)| \le |\pos(x, S_1) - \pos(x, S_{w^*})| + |\pos(x, S_{w^*}) - \rank(x, S)| <  4{w^*} + \delta,
	\]
	and, using again that $G=1$ together with Lemma~\ref{lem:dependent}, we know that
	$\pos(x, S_1)$  only depends on comparisons between elements in $\{ y \in S \, : \, | \pos(x, S_1) - \pos(y, S_{w^*})| \le 6{w^*} \}$ in $S_{w^*}$ which, in turn,
	is a subset of $\{ y \in S \, : \, |pos(x, S_1) - \rank(y, S) | \le 6{w^*} + \delta \}$.
	Combining the previous properties, we have that $pos(x, S_1)$ 
	only depends on $\{ y \in S \, : \, |\rank(x, S) - \rank(y, S) | < 10{w^*} + 2\delta \}$,
	implying that two elements belonging to the same set depend on different comparisons results.

Observe now that each set has size at least $\lfloor \frac{n}{k} \rfloor $ and at most $\lceil \frac{n}{k} \rceil = O(\frac{n}{\log n})$, define
	$D^{(j)} = \sum_{x_i\in P^{(j)}} \disl(x_i,S_1)$ to be the total dislocation of all elements in $P^{(j)}$, and let $\mu^{(j)} =
	\E[{D^{(j)}\mid G=1}] = \sum_{x \in P^{(j)}} \E[\disl(x, S_1) \mid G=1] \le \frac{32 c n}{31 k}$.
	
	We will use Hoeffding's inequality to prove that  $D^{(j)}\le \frac{2 c n}{k}$ with high probability.
	Hoeffding's inequality is as follows: for independent random variables $X_1,\dots,X_n$, such that $X_i$ is in $[a_i,b_i]$, and $X = \sum_i X_i$, 
	\[
		P(X - \E[X]\ge t) \le \exp\left(-\frac{2t^2}{\sum_{i=1}^{n}(b_i-a_i)^2}\right)\, .
	\]
		
	\noindent Hence, since each $\disl(x_i, S_1)$ is between $0$ and $\delta + 4w^* = O(\log n)$ when $G=1$, we have:
\[
	P\left(D^{(j)} - \mu^{(j)} \ge \frac{30 c n}{31 k} \, \Big| \, G=1 \right) \le 
	\exp\left(-\frac{1800}{961} \cdot \frac{\frac{n^2}{c^2 \delta^2}}{ \lceil \frac{n}{k} \rceil (\delta+ 4w^*)^2}\right)
	= \exp \left(-\Omega\left(\frac{n}{\log^3 n}\right)\right).
\]
	
	Finally, by using the union bound over all $k = \Theta(\log n)$ sets and on the event $G=0$, we get that 
	$\disl(S_1) \ge 2 c n$	
	with probability at most $ O(\log n) \cdot  \exp \left(-\Omega\left(\frac{n}{\log^3 n}\right)\right) + \frac{1}{n^5}$, which is at most $\frac{1}{n^4}$ for sufficiently large values of $n$.
\end{proof}

\section{Derandomization}
\label{sec:derand}

\subsection{Partitioning $S$}
\label{sec:derand_partitioning}

In order to run \rifflesort we need to partition the input sequence $S$ into a collection of random sets $T_0, T_1, \dots, T_k$ where $k= \frac{\log n}{2}$ and each $T_i$ contains $m = \sqrt{n} \cdot 2^{i-1}$ elements that are chosen uniformly at random from
the $n - \sqrt{n} \sum_{j=i+1}^k 2^{i-1} = 2 m$ elements in $S \setminus \bigcup_{j=i+1}^k T_j$.
Notice also that this is the only step in the algorithm that is randomized. 
To obtain a version of \rifflesort that does not require any external source of randomness, i.e., that depends only on the input sequence and on the comparison results,
we will generate such a partition by exploiting the intrinsic random nature of the comparison results.

We start by showing that, with probability at least $1-\frac{1}{n^3}$, the partition $T_0, \dots, T_k$ can be found in $O(n)$ time using only $O(n)$ random bits.
To this aim it suffices to show that, given a set $A$ of $2N$ elements, a random set $B \subset A$ of $N$ elements can be 
found in $O(N)$ time using $O(N)$ random bits, with probability at least $1-\frac{1}{N^7}$.
We construct $B$ as follows:
\begin{itemize}
\item For each element of $A$ perform a coin-flip. Let $C$ be the set of all the elements whose corresponding coin flip is ``heads''.

\item If $|C|=N$, return $B = C$. Otherwise, if $|C|<N$, randomly select a set $D$ of $N-|C|$ elements from $A \setminus C$ and return $B = C \cup D$.
Finally, if $|C|>N$, randomly select a set $D$ of $|C|-N$ elements from $C$ and return $B= C \setminus D$.
\end{itemize}

Standard techniques show that this method of selecting $B$ is unbiased, i.e., all the sets $B \subset A$ of $N$ elements are returned with equal probability. We therefore move the proof of the following lemma to the Appendix.
\begin{lemma}
	\label{lemma:random_selection_unbiased}
	For any set $X \subset A$ of $N$ elements, $\Pr(B = X) =\binom{2N}{N}^{-1}$. 
\end{lemma}

We now show an upper bound on the number of random bits required.
\begin{lemma}
	For sufficiently large values of $N$, the number of random bits required to select $B$ u.a.r. is at most $2N$ with probability at least $1 - N^{-7}$.
\end{lemma}
\begin{proof}
	Clearly, $C$ can be built using $3N$ random bits.
    Let $m=|C|$, since $E[|C|]=N$, by Hoeffding's inequality we have:
    \[
    	\Pr(| m - N | \ge 2 \sqrt{N \ln N} ) \le 2 e^{-8 \ln N} \le 2N^{-8}.
    \]
    
    This implies that, with probability at least $1-2N^{-8}$, the set $C$ contains at most $2 \sqrt{N \ln N}$ elements. Hence, $O(\sqrt{N} \cdot \mathrm{polylog} N)$ random bits suffice to select $D$ (using, e.g., a simple rejection strategy), and therefore $B$, with probability at least $1-N^{-8}$. The claim follows by using the union bound.
 \end{proof}

From the above lemmas, it immediately follows that all the sets $T_0, \dots, T_k$ can be constructed in time $O(\sum_{i=0}^k |T_i|) = O(n)$ using at most $4n$ random bits with probability at least $1- n^{-\frac{7}{2}}  \cdot \log N \ge 1 - n^{-3}$, for sufficiently large values of $n$.

\subsection{Derandomized RiffleSort}

	As shown in \cite{newwindowsort}, it is possible to simulate ``almost-fair'' coin flips by xor-ing together a sufficiently large number of comparison results. Indeed, we can associate the two possible results of a comparison with the values $0$ and $1$, so that each comparison behaves as a Bernoulli random variable whose parameter is either $p$ or $1-p$.
	We can then use the following fact: let $c_1, \dots, c_k$ be $k = \Theta(\log n)$ independent Bernoulli random variables such that $P(c_i=1) \in \{p, 1-p\} \; \forall i=1,\dots,k$, then $| P(c_1 \oplus c_2 \oplus \dots \oplus c_k = 0) - \frac{1}{2}| \le \frac{1}{n^4}$.

	Therefore, if we consider the set $A$ containing the first $9 k$ elements from $S$ and we compare each element in $A$ to all the elements in $S \setminus A$, we obtain a collection of $9 k (n - k) \ge 8 k n$ comparison results (for sufficiently large values of $n$) from which we can generate $8n$ almost-fair coin flips.
	With probability at least $1 - \frac{8 k n }{n^4} - n^{-3} > 1 - \frac{1}{n^2}$ these almost-fair coin flips behave exactly as unbiased random bits, and they suffice to select a partition $T_0, \dots, T_k$ of $S \setminus A$.\footnote{This is true even if up to $|S \setminus A| - 1$ additional $+\infty$ elements are added to $S \setminus A$, as described in Section~\ref{sec:rifflesort}}.
	It is now possible to use \rifflesort on $S \setminus A$ to obtain a sequence $S'$ having maximum dislocation $d=O(\log n)$ and linear dislocation $O(n)$ (from Lemma~\ref{lemma:rifflesort_runtime} and Lemma~\ref{lemma:rifflesort_dislocation} this requires time $O(n \log n)$ and succeeds with probability at least $1- |S \setminus A|^{-\frac{3}{2}} > 1 - 3 n^{-\frac{3}{2}}$ since $|S \setminus A| \ge \frac{n}{2}$).
	
	What is left to do is to reinsert all the elements of $A$ into $S'$ without causing any asymptotic increase in the total and in the maximum dislocation. While one might be tempted to use the result of Section~\ref{sec-introduction}, this is not actually possible since the errors between the elements in $A$ and the elements in $S'$ now depend on the permutation $S'$.
	However, a simple (but slower) strategy, which is similar to the one used in \cite{newwindowsort}, works even when the sequence $S'$ is adversarially chosen as a function of the errors, as long as its maximum dislocation is at most $d$.
	Suppose that we have a guess $\tilde{r}$ on $\rank(x, S')$, we can determine whether $\tilde{r}$ is a good estimate on $\rank(x, S')$ 
	by comparing $x$ with all the elements in positions from $\tilde{r} - cd$ to $\tilde{r} + cd -1$ in $S'$ and counting the number $m$ of \emph{mismatches}: a mismatch is an element $y$ such that either (i) $\pos(y, S') < \tilde{r}$ and $y > x$, or (ii) $\pos(y, S') \ge \tilde{r}$ and $y < x$.
	Suppose that our guess of $\tilde{r}$ is much smaller than the true rank of $x$, say $\tilde{r} < \rank(x, S') - c d$ for a sufficiently large constant $c$, then
all the elements $y$ such that $\tilde{r} + d \le \rank(y, S') < \tilde{r} + (c-1)d$
are in $\{ z \in S' \, : \, \tilde{r} \le \pos(z, S') < \tilde{r} + cd \}$.
Since $x \succ y$, we have that the observed comparison result is $x > y$ with probability at least $1-p$, and 
hence the expected number of mismatches $m$ is at least $(c-2)d(1-p)$, and a Chernoff bound can be used to show that with probability $1-\frac{1}{n^4}$,  $m$ will exceed $\frac{1}{2}(c-2)d(1-p) \ge \frac{1}{3} c d$.
A symmetric argument holds for the case in which $\tilde{r} \ge \rank(x, S') + c d$.
On the contrary, if $\rank(x, S') - d \le \tilde{r} < \rank(x, S') + d$, 
all the elements $y$ such that
either $\tilde{r} - (c-2) d \le \rank(y, S') < \tilde{r} -2d$
or $\tilde{r} + 2 d \le \rank(y, S') < \tilde{r} - (c-2) d$
are in the correct relative order w.r.t. $x$ in $S'$.
This implies that the expected number of mismatches $m'$ between $x$ and all the elements $y$ will be at most $2(c-4)d p$, that $m' \le 4(c-4)dp$ with probability at least $1-\frac{1}{n^4}$, which implies that $m \le m' + 4d \le 4(c-4)d p + 4d < \frac{1}{3}c d$ with at least the same probability.

Therefore, to compute a $r_x$ satisfying $|r_x - \rank(x, S')| = O(d)$, it suffices to count the number of mismatches for $\tilde{r} = 0, 2d, 4d, \dots $ and to select the value of $\tilde{r}$ minimizing their number. 
The total time required to to compute all $r_x$ for $x \in A$ is therefore $|A| \cdot O(\frac{n}{d} \cdot d) = O(n \log n)$, and the success probability is at least $1 - O(\frac{n}{d} \cdot |A|) \cdot \frac{1}{n^4} \ge 1 - \frac{1}{n^2}$, for sufficiently large values of $n$.
Combining this with the success probability of \rifflesort, we obtain an overall success probability of at least $1 - \frac{1}{n}$.
	 Finally,  since the set $A$ only contains $O(\log n)$ elements, simultaneously reinserting them in $S'$ affects the maximum dislocation of $S'$ by at most an $O(\log n)$ additive term. Moreover, their combined contribution to the total dislocation is at most $O(\log^2 n)$.

\clearpage
\appendix

\section{Omitted Proofs}

\subsection*{Proof of Lemma~\ref{lemma:draw_no_long_monocolor}}

\begin{proof}[\unskip\nopunct]
Let $b_i$ be the color of the $i$-th drawn ball.

	We consider the sequence of drawn balls and,
	for any position $n$, we bound the distance between $n$-th ball and the position of the $k$-th closest white ball. Suppose that $n \le M = \frac{N}{2}$, as otherwise we can apply similar arguments by considering the drawn balls in reverse order. We distinguish two cases.
	
	If $n \le \frac{M}{4}$ then, for any $j \le \frac{M}{4}$, the probability that $b_{n+j}$ is white, regardless of the colors of the other balls in $b_n, \dots, b_{M/2}$, is at least:
	\[
	\frac{M-(n+j)}{2M-(n+j)}
	= \frac{1}{2} - \frac{n+j}{4M-2(n+j)}
	\ge \frac{1}{2} - \frac{M}{2 \cdot 3M} = \frac{1}{3} > \frac{1}{16}.
	\]
	
	If $\frac{M}{4} \le n \le M$, then let $X$ be the number of white balls in $\{b_1, \dots, b_n \}$.
	Since $X$ is distributed as a hypergeometric random variable of parameters $N$, $M$, and $n$ we have that $\E[X] = \frac{nM}{N} = \frac{n}{2}$.
	
	By using the tail bound (see, e.g., \cite{skala2013hypergeometric}) $\Pr(X \ge \E[X] + tn) \le e^{-2t^2 n}$ for $t \ge 0$, we obtain (for sufficiently large values of $N$):
	\begin{align*}    
	\Pr\left(X \ge \frac{3M}{4}\right) &=
	\Pr\left(X \ge \frac{M}{2} + \frac{M}{4}\right) \le
	\Pr\left(X \ge \frac{n}{2} + 2 \sqrt{n \log n}\right) \\
	&\le e^{-8 \log n} < n^{-8} \le
	2^{24} N^{-8}.
	\end{align*}
	
	Assume now that $X < \frac{3M}{4}$, which happens with probability at least $1- \frac{2^24}{N^8}$.
	In this case, for any $j \le \frac{M}{8}$, the probability that $b_{n+j}$ is white, regardless of the colors of the other balls in $b_n, \dots, b_{(9/8) M}$, is at least:
	\[
	\frac{M-(\frac{3M}{4}+j)}{2M-(n+j)}
	\ge \frac{M-\frac{7M}{8}}{2M}
	= \frac{1}{16}.
	\]
	
	Therefore, in both the first case and in the second case, as long as our assumption holds, the probability that at most $k$ balls in $b_{n}, \dots, b_{n+100k}$ are white is at most:
	\begin{multline*}    
	\sum_{j=0}^{k} \binom{100k}{j} \left( \frac{1}{16} \right)^j \left( \frac{15}{16} \right)^{100k-j} 
	\le (k+1)  \binom{100k}{k} \left(\frac{15}{16} \right)^{100k}   \\
	\le (k+1)  \left(\frac{100ek}{k}\right)^k \left(\frac{15}{16}\right)^{100k} 
	\le (k+1)  \left(100e \left(\frac{15}{16}\right)^{100} \right)^k 
	< \frac{k+1}{2^k},
	\end{multline*}    
	
	where we used the inequality $\binom{\eta}{\kappa} \le \left( \frac{e \eta}{\kappa} \right)^\kappa$.
	
	By using the union bound on all $n \le M$ and on the event  $X < \frac{3M}{4}$ whenever $\frac{M}{4} \le n \le M$, we have can upper bound the sought probability as:
	\[
	N \left( \frac{k+1}{2^k} + 2^{24} N^{-8} \right)
	\le N \left( \frac{N}{N^9} + 2^{24} N^{-8} \right)
	\le  (1+2^{24}) N^{-7} \le N^{-6},
	\]
	where the last inequality holds for sufficiently large values of $N$.
\end{proof}

\subsection*{Proof of Lemma~\ref{lemma:random_selection_unbiased}}
\begin{proof}[\unskip\nopunct]
	For each $k \in [ -N, N ]$, we define a collection $\mathcal{Y}_k$ of sets as follows:
    If $k < 0$, $\mathcal{Y}_k$ contains all the sets $Y \subset X$ such that $| X \setminus Y | = |k|$.
    If $k \ge 0$, $\mathcal{Y}_k$ contains all the sets $Y \supseteq X$ such that $| Y \setminus X | = |k|$.
    Notice that, depending on the value of $k$, $\mathcal{Y}_k$ can be obtained by either selecting $|k|$ elements to remove from $X$, or by selecting $k$ elements to add to $X$ from $A \setminus X$. Therefore, $|\mathcal{Y}_k| = \binom{N}{|k|}$.
    Defining $m=|C|$, we have:
	\begin{align*}
		\Pr(B=X) &= \sum_{k=-N}^N \Pr(m=N+k) \cdot \Pr(B = X \mid m=N+k) \\
  		   &=  2^{-2N} \sum_{k=-N}^N \binom{2N}{N+k} \sum_{Y \in \mathcal{Y}_k} \Pr(C = Y \mid m=N+k) \Pr(B=X \mid C = Y) \\
			&=  2^{-2N} \sum_{k=-N}^N \binom{2N}{N+k} \sum_{Y \in \mathcal{Y}_k} \frac{1}{\binom{2N}{N+k}}  \frac{1}{\binom{N+|k|}{|k|}} 
            =  2^{-2N} \sum_{k=-N}^N \frac{|\mathcal{Y}_k|}{\binom{N+|k|}{|k|}} \\
            &= 2^{-2N} \sum_{k=-N}^N \frac{\binom{N}{|k|}}{\binom{N+|k|}{|k|}} 
			= 2^{-2N} \sum_{k=-N}^N \frac{N! N! |k|!}{ (N+|k|)! (N-|k|)! |k|!} \\
            & = 2^{-2N} \sum_{k=-N}^N \frac{N! N!}{ (N+k)! (N-k)!}
            = 2^{-2N} \sum_{j=0}^{2N} \frac{N! N!}{ j! (2N-j)!}\\
            &= 2^{-2N} \frac{N! N!}{(2N)!} \sum_{j=0}^{2N} \frac{(2N)!}{ j! (2N-j)!} = \frac{N! N!}{(2N)!} = \frac{1}{\binom{2N}{N}}.
	\end{align*}
\end{proof}

\clearpage

\bibliographystyle{plain}
\bibliography{bibliography,references}

\end{document}

%% file: Introduction.tex
\section{Introduction}\label{sec-introduction}
\newcommand{\whp}{w.h.p.}
\newcommand{\expe}{exp.}

We study the problem of \emph{sorting} $n$ distinct elements under \emph{persistent} random comparison \emph{errors}. This  problem arises naturally when sorting is applied to real life scenarios. For example, one could use experts to compare items, with each comparison being performed by one expert. As these operations are typically expensive, one cannot repeat them, and the result may sometimes be erroneous. Still, one would like to reconstruct from these information the correct (or a nearly correct) order of the elements.

In this classical model, which \emph{does not allow resampling}, 
each comparison is wrong with some fixed (small)  probability $p$, and correct with probability $1-p$.\footnote{As in previous works, we assume $p<1/32$ though the results holds for $p<1/16$.} The comparison errors are independent over all possible pairs of elements, but they are persistent: Repeating the same comparison several times is useless since the result is always the same, i.e., always wrong or always correct.

Because of errors, it is impossible to sort correctly and therefore, one seeks to return a ``nearly sorted'' sequence, that is, a sequence where the elements are ``close'' to their correct positions.
To measure the quality of an output sequence in terms of sortedness, a common way is to consider the \emph{dislocation} of an element, which is the difference between its position in the output and its position in the correctly sorted sequence. In particular, one can consider the  \emph{maximum dislocation} of any element in the sequence or the \emph{total dislocation} of the sequence, i.e., the sum of the dislocations of all $n$ elements. 

Note that sorting with persistent errors as above is much more difficult than the case in which comparisons can be repeated, where a trivial $O(n\log^2 n)$ time solution is enough to sort perfectly with high probability (simply repeat each comparison $O(\log n)$ times and take the majority of the results). Instead, in the model with persistent errors, it is impossible to sort perfectly as no algorithm can achieve a maximum dislocation that is smaller than $\Omega(\log n)$ w.h.p., or total dislocation smaller than $\Omega(n)$ in expectation \cite{geissmann_et_al}.
Such a problem has been extensively studied in the literature, and several algorithms have been devised with the goal of sorting \emph{quickly} with small dislocation (see Table~\ref{tb-recurrent}). 
Unfortunately, even though all the algorithms achieve the best possible maximum dislocation of $\Theta(\log n)$, they 
use a truly superlinear number of comparisons (specifically, $\Omega(n^{c})$ with $c\geq 1.5$), and/or require significant amount of time (namely, $O(n^{3+c})$ where $c$ is a big constant that depends on $p$). 
This naturally suggests the following question:
\begin{quote}
	\emph{What is the time complexity of sorting optimally with persistent errors?}
\end{quote}

\noindent In this work, we answer this basic question by showing the following result:
\begin{quote}
	\emph{There exists an algorithm with \textbf{optimal running time} $O(n\log n)$ which  achieves simultaneously \textbf{optimal maximum dislocation} $O(\log n)$ and \textbf{optimal total dislocation} $O(n)$, both \textbf{with high probability}.}
\end{quote}
The dislocation guarantees of our algorithm are optimal, due to the lower bound of \cite{geissmann_et_al}, while the existence of an algorithm achieving a maximum dislocation of $d = O(\log n)$ in time $T(n) = o(n \log n)$ would immediately imply the existence of an algorithm that sorts $n$ elements in $T(n) + O(n \log \log n) = o(n \log n)$ time, even in the absence of comparison errors, thus contradicting the classical $\Omega(n \log n)$ lower bound for comparison-based algorithms.\footnote{Indeed, once the approximately sorted sequence $S$ is computed, it suffices to  apply any $O(n \log n)$ sorting algorithm on the first $m=2\max\{d, \log n\}$ elements of $S$, in order to select the smallest $m/2$ elements. Removing those elements from $S$ and repeating this procedure $\frac{n}{m}$ times, allows to sort in  $T(n) + O(\frac{n}{m} \cdot m \log m)$ time.}

Along the way to our result, we consider the problem of \emph{searching with persistent errors}, defined as follows:
\begin{quote}\emph{We are given an approximately sorted sequence $S$, and an additional element $x \not\in S$. 
	The goal is to compute, under persistent comparison errors, an \emph{approximate rank} (position) of $x$  which differs from the true rank of $x$ in $S$ by a \emph{small} additive error.}
\end{quote}

For this problem, we show an algorithm that requires $O(\log n)$ time to compute, w.h.p., an approximate rank
that differs from the true rank of $x$ by at most $O(\max\{d,\log n\})$, where $d$ is the maximum dislocation of $S$.
For $d=\Omega(\log n)$ this allows to insert $x$ into $S$ without any asymptotic increase of the maximum (and total) dislocation in the resulting sequence. Notice that, if $d$ is also in $O(n^{1-\epsilon})$ for any constant $\epsilon > 0$, this is essentially the best we can hope for, as an easy decision-tree lower bound shows that any algorithm  must require $\Omega(\log n)$ time.
Finally, we remark that \cite{Klein2011} considered the variant in which the original sequence is \emph{sorted}, and the algorithm must compute the correct rank. For this problem, they present an algorithm that runs in $O(\log n \cdot \log \log n)$ time and succeeds with probability $1-f(p)$, with $f(p)$ vanishing as $p$ goes to $0$. As by-product of our result, we can obtain the optimal $O(\log n)$ running time with essentially the same success probability.

\begin{table}
	\centering
	\begin{tabular}{|c|cc|c|}
		\multicolumn{4}{c}{\textbf{Upper bounds}} \\[\smallskipamount] \hline
	\textbf{Running Time}	& \textbf{Max Dislocation} & \textbf{Tot Dislocation} & \textbf{Reference} \\ \hline
	$O(n^{3+c})$	 & $O(\log n)$ \whp & $O(n)$ \whp &  \cite{Braverman2008} \\ 
	$O(n^2)$	& $O(\log n)$ \whp & $O(n\log n)$ \whp  &  \cite{Klein2011} 
	\\ 
	$O(n^2)$	& $O(\log n)$ \whp & $O(n)$ \expe &  \cite{geissmann_et_al} \\ 
	$\tilde{O}(n^{3/2})$	& $O(\log n)$ \whp & $O(n)$ \expe & \cite{newwindowsort} \\ \hline \hline
	$O(n\log n)$	& $O(\log n)$ \whp & $O(n)$ \whp & \textbf{this work}  \\ \hline 
	\multicolumn{4}{c}{} \\[0pt]
	\multicolumn{4}{c}{\textbf{Lower bounds}} \\[\smallskipamount] \hline
	Any & $\Omega(\log n)$ \whp & $\Omega(n)$ \expe &\cite{geissmann_et_al} \\ \hline
	\end{tabular} 
	\caption{The existing approximate sorting algorithms and our result. 
		The constant $c$ in the exponent of the running time of \cite{Braverman2008} depends on the error probability $p$ and it is typically quite large. 
		We write $\Omega(f(n))$ w.h.p. (resp. exp.) to mean that no algorithm can achieve dislocaiton $o(f(n))$ with high probability (resp. in expectation).}
	\label{tb-recurrent}
\end{table}

\subsection{Main Intuition and Techniques}

\paragraph*{Approximate Sorting}

In order to convey the main intuitions behind our $O(n \log n)$-time optimal-dislocation approximate sorting algorithm, we consider the following ideal scenario: 
we already have a perfectly sorted sequence $A$ containing a random half of the elements in our input sequence $S$ and we, somehow, also know the position in which each element $x \in S\setminus A$ should be inserted into $A$ so that the resulting sequence is also sorted (i.e, the \emph{rank} of $x$ in $A$).
If these positions alternate with the elements of $A$, then, to obtain a sorted version of $S$, it suffices to \emph{merge} $S$ and $S \setminus A$, i.e., to simultaneously insert all the elements of $S \setminus A$ into their respective positions of $A$. 
Unfortunately, we are far from this ideal scenario for several reasons: first of all, multiple, say $\delta$, elements in $S \setminus A$ 
might have the same rank in $A$. Since we do not know the order in which those elements should appear, this will already increase the dislocation of the merged sequence to $\Omega(\delta)$. Moreover, due to the lower bound of \cite{geissmann_et_al}, we are not actually able to obtain a perfectly sorted version of $A$ and we are forced to work with a permutation of $A$ having dislocation $d = \Omega(\log n)$, implying that the natural bound on the resulting dislocation can be as large as $d \cdot \delta$. This is a bad news, as one can show that $\delta = \Omega(\log n)$.
However, it turns out that the number of elements in $S \setminus A$ whose positions lie in a $O(\log n)$-wide interval of $A$
is still $O(\log n)$, w.h.p., implying that the final dislocation of $A$ is just $O(\log n)$.

But how do we obtain the approximately sorted sequence $A$ in the first place? We could just recursively apply the above strategy on the (unsorted) elements of $A$, except that this would cause a blow-up in the resulting dislocation due to the constant hidden by the  big-O notation.
We therefore interleave merge steps with invocations of (a modified version of) the sorting algorithm of \cite{geissmann_et_al}, which essentially reduces the dislocation by a constant factor, so that the increase in the worst-case dislocation will be only an \emph{additive} constant per recursive step.

An additional complication is due to the fact that  we are not able to compute the exact ranks in $A$ of the elements in $S\setminus A$. We therefore have to deal, once again, with approximations that are computed using the other main contribution of this paper: \emph{noisy binary search trees}, whose key ideas are described in the following.

\paragraph*{Noisy Binary Search}

As a key ingredient of our approximate sorting algorithm, we need to \emph{merge} an almost-sorted sequence with a set of elements, without causing any substantial increase in the final dislocation.
More precisely, if we are given a sequence $S$ with dislocation $d$ and an element $x$, we want to compute an \emph{approximate rank} of $x$ in $S$, i.e., a position that differs by at most $O(\max\{d, \log n\})$ from the position that $x$ would occupy if the elements $S \cup \{ x \}$ were perfectly sorted.

As a comparison, this same problem has been solved optimally in $O(\log n)$ time in the easier case in which errors are not persistent and $S$ is already sorted \cite{Feige1994}.
The idea of \cite{Feige1994} is to locate the correct position of $x$ using a binary decision tree: ideally each vertex $v$ of the tree
tests whether $x$ appears to belong to a certain \emph{interval} of $S$ and, depending on the results,
one of the children of $v$ is considered next.
As these intervals become narrower as we move from the root towards the leaves, which are in a one-to-one correspondence with positions of $S$, we eventually discover the correct rank of $x$ in $S$.
In order to cope with failures, this process is allowed to \emph{backtrack} when inconsistent comparisons are observed, thus repeating some of the comparisons involving ancestors of $v$. Moreover, 
to guarantee that the result will be correct with high probability, a logarithmic number of consistent comparisons with a leaf are needed before the algorithm terminates.

Notice how this process heavily depends on the fact that it is possible to gather more information on the true relative position of $x$ by repeating a comparison multiple times (in fact, it is easy to design a simple $O(\log^2 n)$-time algorithm by exploiting this fact). Unfortunately, this is not the case anymore when errors are \emph{persistent}.
To overcome this problem we design a \emph{noisy binary search tree} in which the intervals element $x$ is compared with are, in a sense, \emph{dynamic}, i.e., they grow every time the associated vertex is visited.
This, in turn, is a source of other difficulties: first, the intervals of the descendants of $v$ also need to be updated. Moreover, we can obtain inconsistent answers not only due to the erroneous comparisons, but also due to the fact that an interval that initially did not contain $x$ might now become too large. Finally, since intervals overlap, we might end up repeating the same comparison even when two different vertices of the tree are involved.
We overcome these problems by using two search trees that initially comprise of disjoint intervals in $S$, which are selected in a way that ensures that all the bad-behaving vertices are confined into only one of the two trees.

\subsection{Related works}\label{sec:related}
Sorting with \emph{persistent errors} has been studied in several works, starting from  \cite{Braverman2008} who presented the first algorithm achieving optimal dislocation (matching lower bounds appeared only recently in \cite{geissmann_et_al}). 
The algorithm in \cite{Braverman2008} uses only $O(n\log n)$ comparisons, but unfortunately its running time $O(n^{3+c})$ is quite large. For example, for a success probability of $1-1/n$, the analysis in \cite{Braverman2008} yields   $c=\frac{110525}{(1/2-p)^4}$. On the contrary, all subsequent faster algorithms \cite{Klein2011,geissmann_et_al,newwindowsort} -- see Table~\ref{tb-recurrent} --  use a number of comparisons which is asymptotically equal to their respective running time.

Other works considered error models in which repeating comparisons is expensive. For example, \cite{braverman2016parallel} studied algorithms which use a \emph{bounded number of rounds} for some ``easier'' versions of sorting (e.g., distinguishing  the top $k$ elements from the others). In each round, a fresh set of comparison results is generated, and each round consists of $\delta \cdot n$ comparisons. They evaluate the algorithm's performance by estimating the number
of ``misclassified'' elements and also consider a variant in which errors now correspond to missing comparison results.

In general, sorting in presence of errors seems to be computationally more difficult than the error-free counterpart. For instance, \cite{ajtai2016sorting} provides algorithms using \emph{subquadratic} time (and number of comparisons) when errors occur only between  elements whose difference is at most some fixed threshold. Also,  \cite{Damaschke16} gives a \emph{subquadratic} time algorithm when the number $k$ of errors is known in advance. 

As mentioned above, an easier error model is the one with \emph{non-persistent} errors, 
meaning that the same comparison can be \emph{repeated} and the errors are independent, and happen with some probability $p<1/2$. 
In this model it is possible to sort $n$ elements in time $O(n\log(n/q))$, where $1-q$ is the success probability of the algorithm \cite{Feige1994} (see also \cite{alonso,hadji} for the analysis of the classical Quicksort and recursive Mergesort algorithms in this error model).

More generally, computing with errors is often considered in the framework of a two-person game called \emph{R\'{e}nyi-Ulam Game} (see e.g.  the survey \cite{Pelc02} and the monograph \cite{Cicalese13}).

\subsection{Paper Organization}

The paper is organized as follows: in Section~\ref{sec:preliminaries} we give some preliminary definitions; then, in Section~\ref{sec:noisy_binary_search}, we present our noisy binary search algorithm, which will be used in Section~\ref{sec-optimal-sorting} to design a optimal randomized sorting algorithm.
The proof of correctness of this algorithm will make use of an improved analysis of the sorting algorithm of \cite{geissmann_et_al}, which we discuss in Section~\ref{sec:windowsort}. 
Finally, in Section~\ref{sec:derand}, we briefly argue on how our sorting algorithm can be adapted so that it does not require any external source of randomness. Some proofs that only use arguments that are not related to the details of our algorithms are moved to the appendix.

%% file: Sorting.tex
\section{Optimal Sorting Algorithm}\label{sec-optimal-sorting}

\subsection{The algorithm}
\label{sec:rifflesort}

We will present an optimal approximate sorting algorithm that, given a sequence $S$ of $n$ elements, computes, in $O(n\log n)$ worst-case time, a permutation of $S$ having maximum dislocation $O(\log n)$ and total dislocation $O(n)$, w.h.p.
In order to avoid being distracted by roundings, we assume that $n$ is a power of $2$ (this assumption can be easily removed by padding the sequence $S$ with dummy $+\infty$ elements).
Our algorithm will make use of the noisy binary search of Section~\ref{sec:noisy_binary_search} and of algorithm \windowsort  presented in \cite{geissmann_et_al}. For our purposes, we need the following \emph{stronger} version of the original analysis in \cite{geissmann_et_al}, in which the bound on the total dislocation was only given in expectation: 

\begin{restatable}{theorem}{windowsortthm}
\label{thm:windowsort}
	Consider a set of $n$ elements that are subject to random persistent comparison errors.
	For any dislocation $d$, and for any (adversarially chosen) permutation $S$ of these elements whose dislocation is at most $d$, 
	$\windowsort(S,d)$ requires $O(n d)$ wost-case time and computes, with probability at least $1-\frac{1}{n^4}$, a permutation of $S$ having maximum dislocation at most $c_p \cdot \min\{d, \log n \}$ and total dislocation at most $c_p \cdot n$, where $c_p$ is a constant depending only on the error probability $p < \frac{1}{16}$. 
\end{restatable}

We prove this theorem in Section~\ref{sec:windowsort}. Notice that \windowsort also works in a stronger error model in which the input permutation $S$ can be chosen adversarially after the comparison errors
between all pairs of elements have been randomly fixed, as long as the total dislocation of $S$ is at most $d$.

In the remaining of this section, we assume $p<1/32$ in order to be consistent with Section~\ref{sec:noisy_binary_search}, though both the above theorem and the algorithm we are going to present will only require $p<1/16$. 
Based on the binary search in Section~\ref{sec:noisy_binary_search}, we  define an operation that allows us to add a linear number of elements to an almost-sorted sequence without any asymptotic increase in the resulting dislocation, as we will formally prove in the sequel.
More precisely, if $A$ and $B$ are two disjoint subsets of $S$, we denote by $\merge(A,B)$ the sequence obtained as follows:
\begin{itemize}
	\item For each element $x \in B$ compute and index $r_x$ such that $|\rank(s, A) - r_x| \le \alpha d$.
	This can be done using the noisy binary search of Section~\ref{sec:noisy_binary_search}, which succeeds with probability at least $1-\frac{1}{|A|^6}$.
	\item Insert \emph{simultaneously} all the elements $x \in B$ into $A$ in their computed positions $r_x$, breaking ties arbitrarily. Return the resulting sequence.
\end{itemize}

Our sorting algorithm, which we call \rifflesort (see the pseudocode in Algorithm~\ref{alg:rifflesort}), works as follows. 
For $k = \frac{\log n}{2}$, we first partition $S$ into $k+1$ subsets $T_0, T_1, \dots, T_{k}$ as follows. Each $T_i$, with $1 \leq i \leq k$,  contains $2^{i-1} \sqrt{n}$ elements chosen uniformly at random from $S \setminus \{T_{i+1},T_{i+2},\ldots,T_k\}$, 
and $T_0 = S \setminus \{T_{1},T_{2},\ldots,T_k\}$ contains the leftover $n - \sqrt{n} \sum_{i=1}^k 2^{i-1} = \sqrt{n}$ elements.
As its first step, \rifflesort will approximately sort $T_0$ using \windowsort, and then it will alternate merge operations with calls to \windowsort. While, the merge operations allow us to iteratively grow the set of approximately sorted elements to ultimately include all the elements in $S$, each operation also worsens the dislocation by a constant factor.
This is a problem since the rate at which the dislocation increases is faster than the rate at which new elements are inserted. The role of the sorting operations is exactly to circumvent this issue: each \windowsort call locally rearranges the elements, so that all newly inserted elements are now closer to their intended position, resulting in a dislocation increase that is only an additive constant.
The corresponding pseudocode is shown in Algorithm~\ref{alg:rifflesort}, in which $\gamma \ge \max\{202 \alpha, 909 \}$ is an absolute constant.

\begin{algorithm}[t]
	$T_0, T_1, \dots, T_{k} \gets$ partition of $S$ computed as explained in Section~\ref{sec:rifflesort}\;
	$S_0 \gets \windowsort(T_0, \sqrt{n})$\;
	\ForEach{$i=1,\dots,k+1$}
	{
		$S_i \gets \merge(S_{i-1}, T_{i-1})$\;
		$S_i \gets \windowsort(S_i, \gamma \cdot  c_p  \cdot \log n)$\;	
	}
	\Return $S_{k+1}$\;
	\caption{\rifflesort\unskip(S)}
	\label{alg:rifflesort}
\end{algorithm}

\subsection{Analysis}

\begin{lemma}
	\label{lemma:rifflesort_runtime}
	The worst-case running time of Algorithm~\ref{alg:rifflesort} is $O(n \log n)$.
\end{lemma}
\begin{proof}
	Clearly the random partition $T_0, \dots, T_k$ can be computed in time $O(n \log n)$,\footnote{The exact complexity of this steps depends on whether we are allowed to generate a uniformly random integer in a range in $O(1)$ time.
	If this is not the case, then integers can be generated bit-by-bit using rejection. It is possible to show that the total number of required random bits will be at most $O(n)$ with probability at least $1-n^{-2}$ (see Section~\ref{sec:derand_partitioning}). To maintain a worst-case upper bound on the running time also in the unlikely event that $O(n \log n)$ bits do not suffice, we can simply stop the algorithm and return any arbitrary permutation of $S$. This will not affect the high-probability bounds in presented in the sequel.} and the first call to \windowsort requires time $O(|T_0| \cdot \sqrt{n}) = O(n)$ (see Theorem~\ref{thm:windowsort}). We can therefore restrict our attention to the generic $i$-th iteration of the for loop.
	The call to $\merge(S_{i-1}, T_{i-1})$ can be performed in $O(|S_i| \log n)$ time since, for each $x \in T_{i-1}$, the required approximation of 
$\rank(x, S_{i-1})$ can be computed in time $O(\log |T_{i-1}|)$ and $|T_{i-1}|< |S_i| \le n$, while inserting the elements in $S_{i-1}$ their computed ranks requires linear time in $|S_{i-1}| + |T_{i-1}| = |S_i|$.
	The subsequent execution of \windowsort with $d = O(\log n)$ requires time $O(|S_i| \log n)$, where the hidden constant does not depend on $i$.	
	Therefore, for a suitable constant $c$, the time spent in the $i$-th iteration is $c |S_i| \log n$ and total running time of Algorithm~\ref{alg:rifflesort} can be upper bounded by:
	\[
		c \sum_{i=1}^{k+1} |S_i| \log n =  c \sqrt{n} \log n \cdot \sum_{i=1}^{k+1} 2^i 
		< 2^{k+2} c \sqrt{n} \log n  		
		= 2 c n \log n.
	\]
	This completes the proof.
\end{proof}

The following lemma, that concerns a thought experiment involving urns and randomly drawn balls, is instrumental to bounding the dislocation of the sequences returned by the $\merge$ operations. Since it can be proved using arguments that do not depend on the details of \rifflesort, we postpone its proof to the appendix. 

\begin{lemma}
	\label{lemma:draw_no_long_monocolor}
	Consider an urn containing $N=2M$ balls, $M$ of which are white and the remaining $N$ are black.
	Balls are iteratively drawn from the urn without replacement until the urn is empty. 
	If $N$ is sufficiently large and $ 9 \log N \le k \le \frac{N}{16}$ holds, the probability that any contiguous subsequence of at most $100k$ drawn balls contains $k$ or fewer white balls is at most $N^{-6}$.
\end{lemma}

We can now show that, if $A$ and $B$ contain randomly selected elements, the sequence returned by  $\merge(A,B)$ is likely to have a dislocation that is at most a constant factor larger than the dislocation of $A$:
\begin{lemma}
\label{lemma:merge_constant_disl_increase}
Let $A$ be a sequence containing $m$ randomly chosen elements from $S$ and having maximum dislocation at most $d$, with $\log n \le d = o(m)$.
Let $B$ be a set of $m$ randomly chosen elements from $S \setminus A$.
Then, for a suitable constant $\gamma$, and for large enough values of $m$, $merge(A,B)$ has maximum dislocation at most $\gamma d$ with probability at least $1-m^{-4}$.
\end{lemma}
\begin{proof}
	Let $\beta = \max\{\alpha, 9/2 \}$, $S'=\merge(A,B)$, and $S^* = \langle s_0^*, s_1^*, \dots, s_{2n-1}^* \rangle$ be the sequence obtained by sorting $S'$ according to the true order of its elements.
	Assume that:
	\begin{itemize}
		\item all the approximate ranks $r_x$, for $x \in B$, are such that $|r_x - \rank(x,A)| \le \beta d$; and
		\item all the contiguous subsequences of $S^*$ containing up to $2 \beta d+2$ elements in $A$ have length at most $200 \beta d + 200$. 
	\end{itemize}	 	
	We will show in the sequel that the above assumptions are likely to hold.
	
	Pick any element $x \in S'$. 
	We will show that our assumptions imply that the dislocation of $x$ in $S'$ is at most $201d$.
	
	An element  $y \in B$ can affect the final dislocation of $x$ in $S'$ only if one of the following two (mutually exclusive) conditions holds: (i) $y \prec x$ and $r_y \ge r_x$, or (ii) $y \succ x$ and $r_y \le r_x$.
	All the remaining elements in $B$ will be placed in the correct relative order w.r.t.\ $x$ in $S'$, and hence they do not affect the final dislocation of $x$.
	
	If (i) holds, we have:
	\[
	r_x - \beta d \le r_y - \beta d \le \rank(y, A) \le \rank(x, A) \le r_x + \beta d,
	\]
	while, if (ii) holds, we have:
	\[
	r_x -\beta d \le \rank(x, A) \le \rank(y, A) \le r_y +  \beta d \le r_x + \beta d,
	\]
	and hence, all the elements $y \in B$ that can affect the dislocation of $x$ in $S'$ are contained in the set
	$Y = \{ y \in B :  r_x - \beta d \le \rank(y, A) \le r_x + \beta d\}$.
	
	We now upper bound the cardinality of $Y$.
	Let $y^-$ be the $(r_x - \beta d - 1)$-th element of $A$;
	if no such element exists, then let $y^- = s^*_0$.
	Similarly, let $y^+$ be the $(r_x +  \beta d)$-th element of $A$; if no such element exists, then let $y^+ = s^*_{2m-1}$.
	Due to our choice of $y^-$ and $y^+$ we have that $\forall y \in Y, 
	y^- \preceq y \preceq y^+$, implying that all the elements in $Y$ appear in the contiguous subsequence $\overline{S}$ of $S^*$ having $y^-$ and $y^+$ as its endpoints. 
	Since no more than $2 \beta d+2$ elements of $A$ belong to $\overline{S}$ , our assumption guarantees that $\overline{S}$ contains at most $200 \beta d+200$ elements.
	This implies that the dislocation of $x$ in $S'$ is at most $\beta d + |Y| \le \beta d + |\overline{S}| \le 201 \beta d + 200 \le \gamma d$, where the last inequality holds for large enough $n$ once we choose $\gamma = 202 \beta$.

	To conclude the proof we need to show that our assumptions holds with probability at least $1-|S'|^{-6}$.
	Regarding the first assumption, for $x\in B$, a noisy binary search returns a rank $r_x$ such that
	$|r_x - \rank(x,A)| \le \alpha d \le \beta d$ with probability at least $1 - O(\frac{1}{m^6})$. Therefore the probability that the assumption holds is at least $1-O(\frac{1}{m^5})$.
	
	Regarding our second assumption, notice that, since the elements in $A$ and $B$ are randomly selected from $S$, we can relate their distribution in $S^*$
		with the distribution of the drawn balls in the urn experiment of Lemma~\ref{lemma:draw_no_long_monocolor}: the urn contains $N=2m$ balls each corresponding to an elements in $A \cup B$, a ball is white if it corresponds to one of the $M=m$ elements of $A$, while a black ball corresponds one of the $M=m$ elements of $B$.	
	If the assumption does not hold, then there exists a contiguous subsequence of $S^*$ of at least $200 \beta d+200$ elements that contains at most $2 \beta d+2$ elements from $A$. By Lemma~\ref{lemma:draw_no_long_monocolor} with $k=2 \beta d+2$, this happens with probability at most $(2m)^{-6}$ (for sufficiently large values of $n$).
	The claim follows by using the union bound.
\end{proof}

We can now use Lemma~\ref{lemma:merge_constant_disl_increase} and Theorem~\ref{thm:windowsort} together to derive an upper bound to the final dislocation of the sequence returned by Algorithm~\ref{alg:rifflesort}.

\begin{lemma}
	\label{lemma:rifflesort_dislocation}
	The sequence returned by Algorithm~\ref{alg:rifflesort} has maximum dislocation $O(\log n)$ and total dislocation $O(n)$ with probability $1- \frac{1}{n\sqrt{n}}$.
\end{lemma}
\begin{proof}
For $i = 1, \dots, k+1$, we say that the $i$-th iteration of Algorithm~\ref{alg:rifflesort} is \emph{good} if the sequence $S_i$ computed at its end
has both (i) maximum dislocation at most $c_p \log n$, and (ii) total dislocation at most $c_p |S_i|$.
As a corner case, we say that iteration $0$ is good if the sequence $S_0$ also satisfies conditions (i) and (ii) above.

We now focus on a generic iteration $i\ge 1$  and show that, assuming that iteration $i-1$ is good, iteration $i$ is also good with probability at least $1-\frac{1}{n^2}$.
Since iteration $0$ is good with probability at least $1-\frac{1}{|S_0|^4}$ = $1- \frac{1}{n^2}$ and there are $k + 1 = O(\log n)$ other iterations, the claim will follow by using the union bound.

Since iteration $i-1$ was good, the sequence $S_{i-1}$ has maximum dislocation $c_p \log n$
and hence the sequence resulting from call to $\merge(S_{i-1}, T_{i-1})$ returns a sequence with dislocation at most $\gamma c_p \log n$ with probability at least $1-\frac{1}{|T_{i-1}|^4} \ge 1- \frac{1}{n^2}$.
If this is indeed the case, we have that the sequence $S_i$ returned by the subsequent call to \windowsort satisfies (i) and (ii) with probability at least $1 - \frac{1}{|S_{i+1}|^4} \ge 1 - \frac{1}{n^2}$. The claim follows by using the union bound and by noticing that the returned sequence is exactly $S_{k+1}$.
\end{proof}

We have therefore proved the following result, which follows directly from Lemma~\ref{lemma:rifflesort_dislocation} and Lemma~\ref{lemma:rifflesort_runtime}:
\begin{theorem}
	\rifflesort is a randomized algorithm that approximately sorts, in $O(n \log n)$  worst-case time, $n$ elements subject to random persistent comparison errors so that the maximum (resp. total) dislocation of the resulting sequence	is $O(\log n)$ (resp. $O(n)$), w.h.p.
\end{theorem}